%% file: tr.tex
\pdfoutput=1 
\documentclass[a4paper,12pt]{article}
\usepackage{fullpage}
\usepackage{authblk}
\usepackage{amsfonts}

\usepackage{subfigure}
\usepackage{lscape}

\usepackage{verbatim}

\usepackage{url}
\usepackage{color}
\usepackage{epic}
\usepackage{eepic}
\usepackage{amsmath}
\usepackage{amsthm}
\usepackage{graphicx}
\input{makros2.tex}
\newcommand{\postponed}[1]{}
\renewcommand{\frage}[1]{}

\newtheorem{theorem}{Theorem}

\newtheorem{corollary}[theorem]{Corollary}

\title{MultiQueues: Simpler, Faster, and Better\\ Relaxed Concurrent Priority Queues}
\author[1]{Hamza Rihani\thanks{\tt hamza.rihani@outlook.com}}
\author[2]{Peter Sanders\thanks{{\tt sanders@kit.edu}}}
\author[3]{Roman Dementiev\thanks{{\tt roman.dementiev@intel.com}}}
\affil[1]{Universit\'e Joseph Fourier Grenoble, France}
\affil[2]{Karlsruhe Institute of Technology, Germany}
\affil[3]{Intel GmbH Munich, Germany}

\begin{document}
\maketitle
\pagestyle{plain}

\begin{abstract}
Priority queues with parallel access are an attractive data structure for applications like prioritized online scheduling, discrete event simulation, or branch-and-bound. However, a classical priority queue constitutes a severe bottleneck in this context, leading to very small throughput. Hence, there has been significant interest in concurrent priority queues with a somewhat relaxed semantics where deleted elements only need to be close to the minimum. In this paper
we present a very simple approach based on multiple sequential priority queues.
It turns out to outperform previous more complicated data structures while at the same time improving the quality of the returned elements.
\end{abstract}

\section{Introduction}
\label{s:intro}

Priority queues are a fundamental data structure with many applications.
Priority queues manage a set of elements and support the operations for inserting elements and deleting the smallest element (\Id{deleteMin}).
Whenever we have to dynamically reorder operations performed by an algorithm, priority queues can turn out to be useful. Examples include graph algorithms for shortest paths and minimum spanning trees, discrete event simulation, best first branch-and-bound, and other best first heuristics. In particular, many successful job scheduling algorithms are of that type. The latter application was the starting point for this work focusing on scheduling of jobs with priorities in data management systems to support quality of service guarantees (Service Level Agreements).

On modern parallel hardware, we often have the situation that $p$ parallel threads (or PEs for \emph{processing   elements}) want to access the priority queue concurrently. This is problematic for several reasons. First of all, even the semantics of a parallel priority queue are unclear. For example, when one PE wants to delete the smallest element, there may be several other PEs in the process of inserting elements. Some of these elements may be smaller than the current minimum. For fundamental physical reasons, there is no way that knowledge about the smallest element can instantaneously propagate through the system. We may try to repair that situation by serializing access to a centralized priority queue, e.g., by locking it. However, this way the priority queue constitutes a serious performance bottleneck. A more attractive way out is to admit that we cannot hope to delete the absolute smallest element and relax the semantics of the \Id{deleteMin} operation allowing it to return also nonminimal elements. These \emph{relaxed priority queues} should support high operation throughput while deleting elements close to the global minimum. 

Consequently, there has recently been intensive work on relaxed priority queues \cite{AKLS14,WVTCT14,HKPSS13}. These works focus on SkipLists and other centralized data structures.  This was surprising to us since in a previous work we experienced that sequential SkipLists are considerably more expensive than classical search trees \cite{DKMS04} which are in turn more expensive than specialized priority queues like heaps. Moreover, previous works on parallel priority queues suggest that a relaxed priority queue could be built from multiple sequential priority queues \cite{DeoPra92,RanEtAl94,San98a}. See Section~\ref{s:related} for a more detailed discussion. In this paper, we present MultiQueue, a simple relaxed priority queue based on these observations. 

The idea for the resulting MultiQueue data structure is very simple.
If there are $p$ parallel threads in the system, we manage an array of $cp$ sequential priority queues for some constant $c>1$. By having considerably more queues than threads, we ensure that contention remains small. Insertions go to random queues. The idea behind this is that each queue should contain a representative sample of the globally present elements. Deleted elements are the minimum of the minima of \emph{two} randomly chosen queues. By choosing two rather than one queue, fluctuations in the distribution of queue elements are stabilized. Refer to Section~\ref{s:multi} for a more detailed discussion as well as Section~\ref{s:implementation} for further implementation details.

Experiments in Section~\ref{s:experiments} indicate that MultiQueues scale excellently on a single socket and reasonably well on a two-socket system. 
In particular, MultiQueues outperforms previous implementations by a factor of two or larger. Moroever, the quality of the deleted elements are higher, i.e., they are closer to the global minimum.
Section~\ref{s:conclusions} summarizes the results and outlines possible further improvements of the MultiQueue idea.

\section{Preliminaries}
\label{s:prelim}

A priority queue $Q$ represents a set of elements. We use $n=|Q|$ for the size of the queue. Classical priority queues support the operations \Id{insert} for inserting an element and \Id{deleteMin} for removing the smallest element. 
The most frequently used sequential priority queue is the binary heap \cite{Wil64}. For large queues, more cache efficient data structures are preferred
\cite{LaMLad97a,San00b}. A simple measure is to increase the degree of the underlying heap and ensure that all successors of a node in the heap reside in the same cache line \cite{LaMLad97a}.

A relaxed priority queue does not require a \Id{deleteMin} operation to return the minimum element. A natural quality criterion is the \emph{rank error} of the deleted elements among all the elements in the queue, i.e., how many elements are smaller than the deleted element. Over the entire use of the queue, one can look at the mean rank, the largest observed rank, or, more generally, the distribution of observed ranks.

\section{More Related Work}
\label{s:related}

There has been considerable work on parallel priority queues in the 1990s \cite{DeoPra92,RanEtAl94,San98a}. 
These queues differ from relaxed priority queues in that they assume synchronized batched operation but they are still relevant as a source of ideas for asynchronous implementations. Indeed, Sanders \cite{San98a} already discusses how these data structures could be made asynchronous in principle: Queue server PEs periodically and collectively extract the globally smallest elements from the queue moving them into a buffer that can be accessed asynchronously. Note that within this buffer, priorities can be ignored since all buffered elements have a low rank. Similarly, an insertion buffer keeps recently inserted elements.
Moreover, the best theoretical  results on parallel priority queues give us an idea how well relaxed priority queues should scale asymptotically. For example, Sanders' queue \cite{San98a} removes the $p$ smallest elements of the queue in time $\Oh{\log n}$. This indicates that worst case rank error \emph{linear} in the number of PEs should be achievable. Current relaxed heaps are a polylogarithmic factor away from this even on the average.

Sanders' queue \cite{San98a} is based on the very simple idea to maintain a local priority queue on each PE and to send inserted elements to random PEs. This idea is actually older, stemming from Karp and Zhang \cite{KarZha93}.
This result could actually be viewed as a relaxed priority queue. Elements are inserted into the queue of a randomly chosen PE. Each PE deletes elements from its local queue. It is shown that this approach leads to only a constant factor more work compared to a sequential algorithm for a branch-and-bound application 
where producing and consuming elements takes constant time. Unfortunately, for a general relaxed priority queue, the Karp Zhang queue \cite{KarZha93} has limitations since slow PEs could ``sit'' on small elements while fast PEs would busily process elements with high rank -- in the worst case, the rank error could grow arbitrarily large.
The MultiQueue we introduce in Section~\ref{s:multi} builds on the Karp Zhang queue \cite{KarZha93}, adapting it to a shared memory setting, decoupling the number of queues from the number of threads, and, most importantly, using a different, more robust protocol for \Id{deleteMin}.

Many previous concurrent priority queues are based on the SkipList data structure \cite{Pug90}. At its bottom, the SkipList is a sorted linked list of elements. Search is accelerated by additional layers of linked lists.
Each list in level $i$ is a random sample of the list in level $i-1$.
Most previous concurrent priority queues delete the exact smallest element \cite{ShaLot00,SunTsi03,LinJon13,CMH14}. This works well if there are not too many concurrent \Id{deleteMin} operations competing for deleting the same elements. However, this inevitably results in heavy contention if very high throughput is required.
The SprayList \cite{AKLS14} reduces contention for \Id{deleteMin} by navigating not to the global minimum but to an element among the $\Oh{p\log^3 p}$ smallest element and deleting it. However, for worst case inputs, insertions can still cause heavy contention. This is a fundamental problem of any data structure that attempts to maintain a single globally sorted sequence. 
Wimmer et al. \cite{WVTCT14} describe a relaxed priority queue for task scheduling based on a hybrid between local and global linked lists and local priority queues. Measurements in Alistarh et al. \cite{AKLS14} indicate that this data structure does not scale as well as SprayLists -- probably due to a frequently accessed central linked list.
Henzinger et al. \cite{HKPSS13} give a formal specification of relaxed priority queues and mention a SkipList based implementation without giving details.
Interestingly, for a relaxed FIFO-queue, the same group proposes a MultiQueue-like structure \cite{HLHPSKS13}.

\section{MultiQueues}
\label{s:multi}

\begin{figure}[b]
\begin{code}
\Procedure insert$(e)$\+\\
  \Repeat\+\\
    $i\Is \Id{uniformRandom}(1..p)$\\
    try to lock $Q[i]$\RRem{e.g. a CAS instruction}\-\\
  \Until lock was successful\\
  $Q[i].\Id{insert}(e)$\\
  $Q[i].\Id{lock}\Is 0$\RRem{unlock}\\
\end{code}
\caption{\label{alg:insertMQ}Insertion into a MultiQueue.}
\end{figure}
\begin{figure}[t]
\begin{code}
\Procedure deleteMin\+\\
  \Repeat\+\\
    $i\Is \Id{uniformRandom}(1..p)$\\
    $j\Is \Id{uniformRandom}(1..p)$\\
    \If $Q[i].\Id{min}>Q[j].\Id{min}$ \Then swap $i$, $j$\\
    try to lock $Q[i]$\RRem{e.g. a CAS instruction}\-\\
  \Until lock was successful\\
  $e\Is Q[i].\Id{deleteMin}$\\
  $Q[i].\Id{lock}\Is 0$\RRem{unlock}\\
  \Return $e$
\end{code}
\caption{\label{alg:deleteMinMQ}\Id{DeleteMin} from a MultiQueue.}
\end{figure}

Our MultiQueue data structure is an array $Q$ of $cp$ sequential priority queues where $c$ is a tuning parameter and $p$ is the number of parallel threads. Access to each queue is protected by a lock flag. \Id{Insert} locks a random unlocked queue $Q[i]$ and inserts the element into $Q[i]$, see Figure~\ref{alg:insertMQ} for pseudocode. Note that this operation is wait-free since we never wait for a locked queue. Since at most $p$ queues can be locked at any time, for $c>1$ we will have constant success probability. Hence, the expected time for locking a queue is constant. Together with the time for insertion we get expected insertion time $\Oh{\log n}$.

A simple implementation of \Id{deleteMin} could look very similar -- lock a random unlocked queue and return its minimal element. This is indeed what we tried first. However, the quality of this approach leaves a lot to be desired. In particular, quality deteriorates not only with $p$ but also with the queue size. One can show that the rank error grows proportional to $\sqrt{n}$ due to random
fluctuations in the number of operations addressing the individual queues.
Therefore we invest slightly more effort into a \Id{deleteMin} by looking at \emph{two} random queues and deleting from the one with smaller minimum. See Figure~\ref{alg:deleteMinMQ} for a simple pseudocode and refer to Section~\ref{s:implementation} for other implementation options.  Our intuition why considering two choices 
may be useful, stems from previous work on randomized load balancing, where it is known that placing a ball on the least loaded of two randomly chosen machine gives a maximum load that is very close to the average load independent of the number of allocated balls \cite{BCSV00}.
The effect of an additional choice on execution time is modest -- we still need only constant expected time to lock the queue we want to access and then spend logarithmic time for the local \Id{deleteMin}. We get the following result on 
running time:

\begin{theorem}
For a MultiQueue with $c>1$, the expected execution time of operations \Id{insert} and \Id{deleteMin} is $\Oh{1}$ plus the time for the sequential queue access.
\end{theorem}
\begin{proof}
The claim follows immediately from what we said above. This even holds in a realistic model for the time of an atomic memory access where
the access time is proportional to the number of contending PEs --
the expected contention is constant due to randomization.
\end{proof}

An advantage about our wait-free locking is that there is no need for complicated lock-free data structures for the individual queues -- any known sequential data structure can be used. The most obvious choice is to use binary heaps, or, more generally $d$-ary heaps.

\begin{corollary}
MultiQueues with $d$-ary heaps need constant average insertion time and 
expected time $\Oh{\log n}$ per operation for worst case operations sequences.
\end{corollary}
 
An interesting example for a different sequential queue is van Emde Boas trees \cite{Emde77}
which allow us to exploit the structure of integer keys:

\begin{corollary}
MultiQueues with van Emde Boas trees need expected time $\Oh{\log\log U}$ per operation for worst case operations sequences with integer keys in $\set{0,\ldots,U}$.
\end{corollary}
Note that van Emde Boas trees can also be practical \cite{DKMS04}.

A note is in order on handling empty queues. If we implement the individual queues so that they contain a sentinel element with key $-\infty$, then a result of $-\infty$ indicates that the overall data structure contains only few elements and may actually be empty. Our experiments simply repeat until an element is found since there it is known that the queue is nonempty. It is an interesting question how to actually and reliably detect a globally empty queue. However, as far as we know most concurrent priority queues, even exact ones, have some issues there. We are not addressing this problem here since we believe that a proper termination detection protocol is often a problem of the application rather than of the priority queue data structure. For example, suppose we are running Dijkstra's shortest path algorithm and some protocol determines that the queue of nodes to be scanned is empty. Then we still cannot stop the search since there might be a thread currently scanning a node which will later insert new elements into the queue.

\subsection{Quality Analysis}\label{ss:analysis}

Unfortunately, we do not have a closed form analysis of the quality of the MultiQueue. However, with simplifying assumptions, we can get a reasonable approximation of what to expect. 

So let us assume for now that all remaining $m$ elements have been allocated uniformly at random to the local queues. This is true when there have been no \Id{deleteMin} operations so far (and the open question is whether the \Id{deleteMin} operations steer the system state away from this situation and whether this is good or bad for quality). Furthermore, let us assume that no queue is locked. 

With these assumptions, the probability to delete an element of rank $i$ is 
$$\prob{\mathrm{rank}=i}=\left(1-\frac{2}{cp}\right)^{i-1}\cdot\frac{2}{cp}$$
The first factor expresses that the $i-1$ elements with smaller ranks are not present at the two chosen queues and the second factor is the probability that the element with rank $i$ \emph{is present}. 
Therefore, the expected rank error in the described situation is  
\begin{equation}\label{eq:eRank}\sum_{i=1}^mi\prob{\mathrm{rank}=i}\leq
  \sum_{i\geq 1}i\left(1-\frac{2}{cp}\right)^{i-1}\cdot\frac{2}{cp}=\frac{c}{2}p
\end{equation}
i.e., linear in $p$.

We can also compute the the cumulative probability that the rank of the deleted element is larger than $k$. This happens when none of the $k$ elements with rank $\leq k$ are present on the two chosen queues. 
\begin{equation}\label{eq:tail}
\prob{\mathrm{rank}>k}=\left(1-\frac{2}{cp}\right)^k\punkt
\end{equation}
$\prob{\mathrm{rank}>k}$ drops to 
$p^{-a}$ for $k=\frac{ca}{2}p\ln p$, i.e., with probability polynomially large in $p$, we have rank error  $\Oh{p\log p}$.

We can also give qualitative arguments how the performed operations change the distribution of the elements. Insertions are random and hence move the system towards a random distribution of elements. DeleteMins are more complicated. 
However, they tend to remove more elements from queues with small keys than from queues with large keys, thus stabilizing the system.

\section{Implementation Details}
\label{s:implementation}

Even when the queue is small, cache efficiency is a major issue for MultiQueues since an access to a queue $Q[i]$ by a PE $j$ will move the cache lines accessed by the operation into the cache of PE $j$. But most likely, some random other PE $j'$ will next need that data causing not only cache misses for $j'$ but also invalidation traffic for $j$. Hence, we need a priority queue whose operations touch only a small number of cache lines for each operation. Our current implementation uses a 8-ary heap from the Boost library ({\tt boost::heap::d\_ary\_heap} \cite{schling2011boost}). This yields up to three times less cache misses than binary heaps and still guarantees efficient worst case access time. This property is also relevant for MultiQueues since an operation that exceptionally takes very long would lock a queue for a long time.  If this queue also contains low rank elements, the rank error of elements deleted in the mean time could become large. In this respect, resizing operations for array based heaps could be a problem. Our implementation allocates enough memory so that resizes are unlikely. However, a more robust implementation might want to preallocate a large amount of virtual memory for each queue. This can be done in such a way that only the actually needed parts are allocated physically \cite{SanWas11}. 

Our \Id{deleteMin} operation samples random queues until it has found two unlocked ones. In order to be able to safely inspect the minimum of a queue without locking it, the minimum of $Q[i]$ is stored redundantly in the same cache line as the lock variable. Updates of the underlying $d$-ary heap update this value whenever necessary. The pseudocode in Figure~\ref{alg:deleteMinMQ} and our implementation tolerates the possibility that the minimum element of $Q[i]$ is deleted immediately before $Q[i]$ is actually locked. In this case, a different, larger element would actually be deleted. This affects quality but does not undermine the correctness of the algorithm. However, we expect that this situation is unlikely in practice. If desired, one can also eliminate this possibility and retry after finding that the locked queue has changed its minimum.  Using a slightly more complicated implementation, one could save a few probes to queues: Rather than looking for two fresh queues when $Q[i]$ turns out to be locked, one could stick to queue $j$ and look for just one fresh queue.

For MultiQueue insertion (see Figure~\ref{alg:insertMQ}) we can avoid some atomic memory access instructions by first reading the lock value and only trying to lock it when the value is 0. However, in our implementation this was only useful for very small values of $c$ so that we immediately try to lock.

\section{Experiments}
\label{s:experiments}

Experiments were performed on a dual socket system with Intel\textsuperscript{\textregistered} Xeon\textsuperscript{\textregistered} CPU E5-2697 v3, 2.60 GHz processors (Haswell). Each socket has 14 cores with two hardware threads, i.e., 56 overall.

We use GCC 4.8.2 with optimization level {\tt -O3}, Boost version 1.56, and Posix threads for parallelization. 

Queue elements are key-value-pairs consisting of two 32 bit integers. 
Initially, the queues are filled with $n_0$ elements with keys uniformly distributed in $\set{0,\ldots,10^8}$ When not otherwise specified, $n_0=10^6$.
Our measurements use threads alternating between \Id{insert} and \Id{deleteMin} operations. Each thread is pinned to a logical core.
Insertions insert elements with keys uniformly selected in $\set{0,\ldots,10^8}$.

\subsection{Throughput}
\begin{figure}
\centering\includegraphics[width=0.8\linewidth]{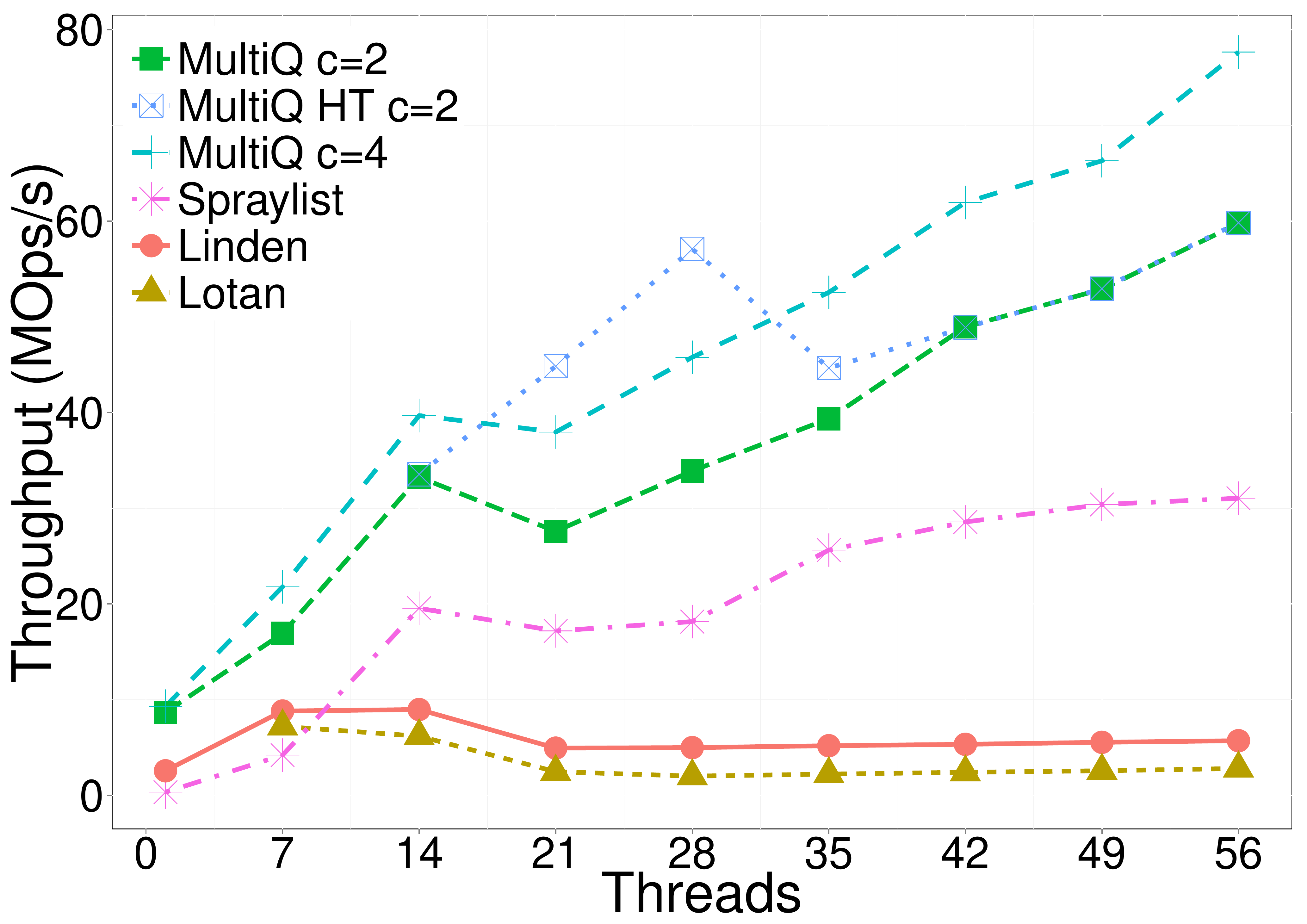}
\caption{Throughput of 50\% \Id{insert} 50\% \Id{deleteMin} operations of uniformly distributed keys} 
\label{fig:throughput}
\end{figure}

\begin{figure}
\centering\includegraphics[width=0.8\linewidth]{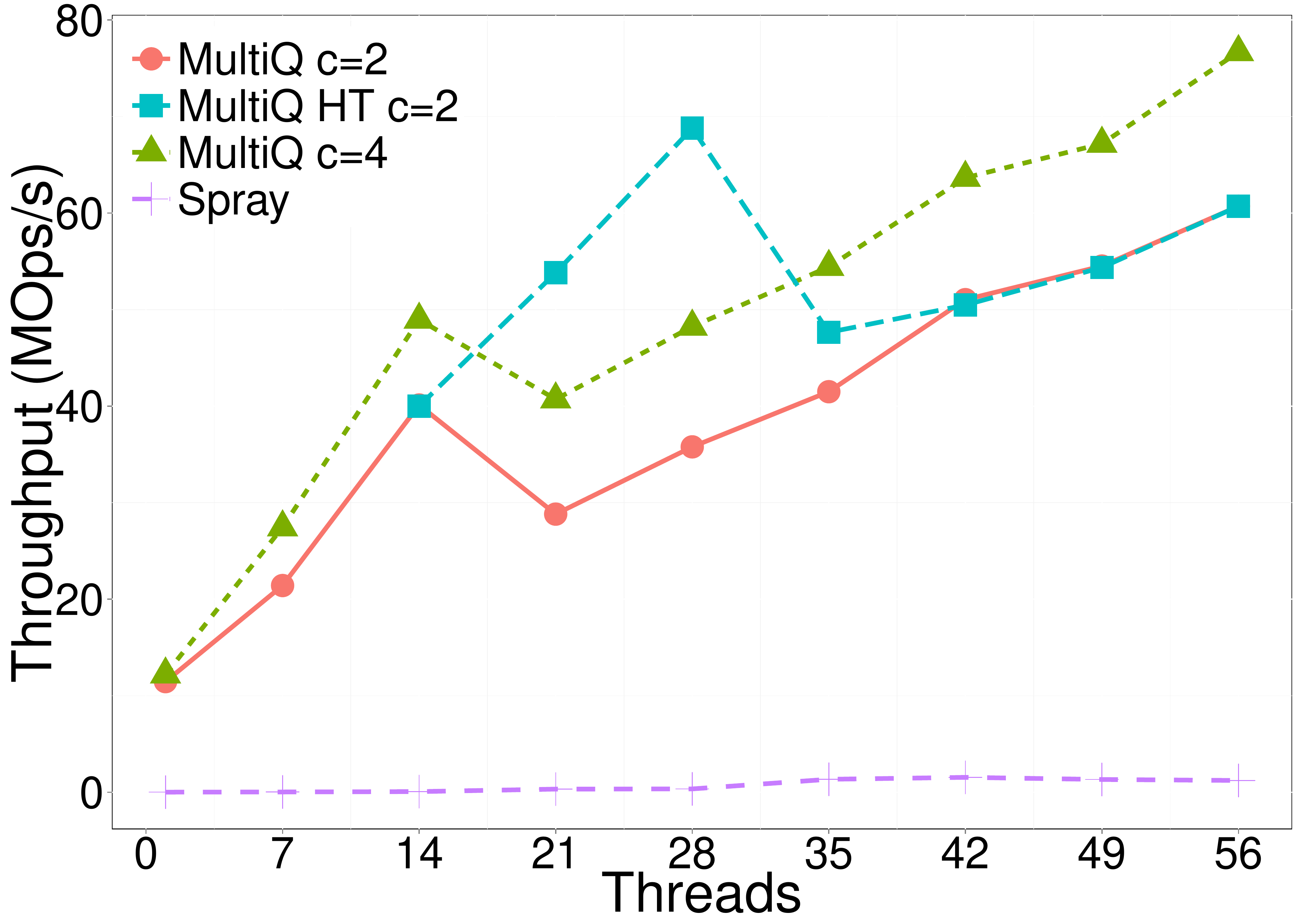}
\caption{Throughput of 50\% \Id{insert} 50\% \Id{deleteMin} operations of monotonic keys}
\label{fig:mono}
\end{figure}

\begin{figure}
\centering\includegraphics[width=0.8\linewidth]{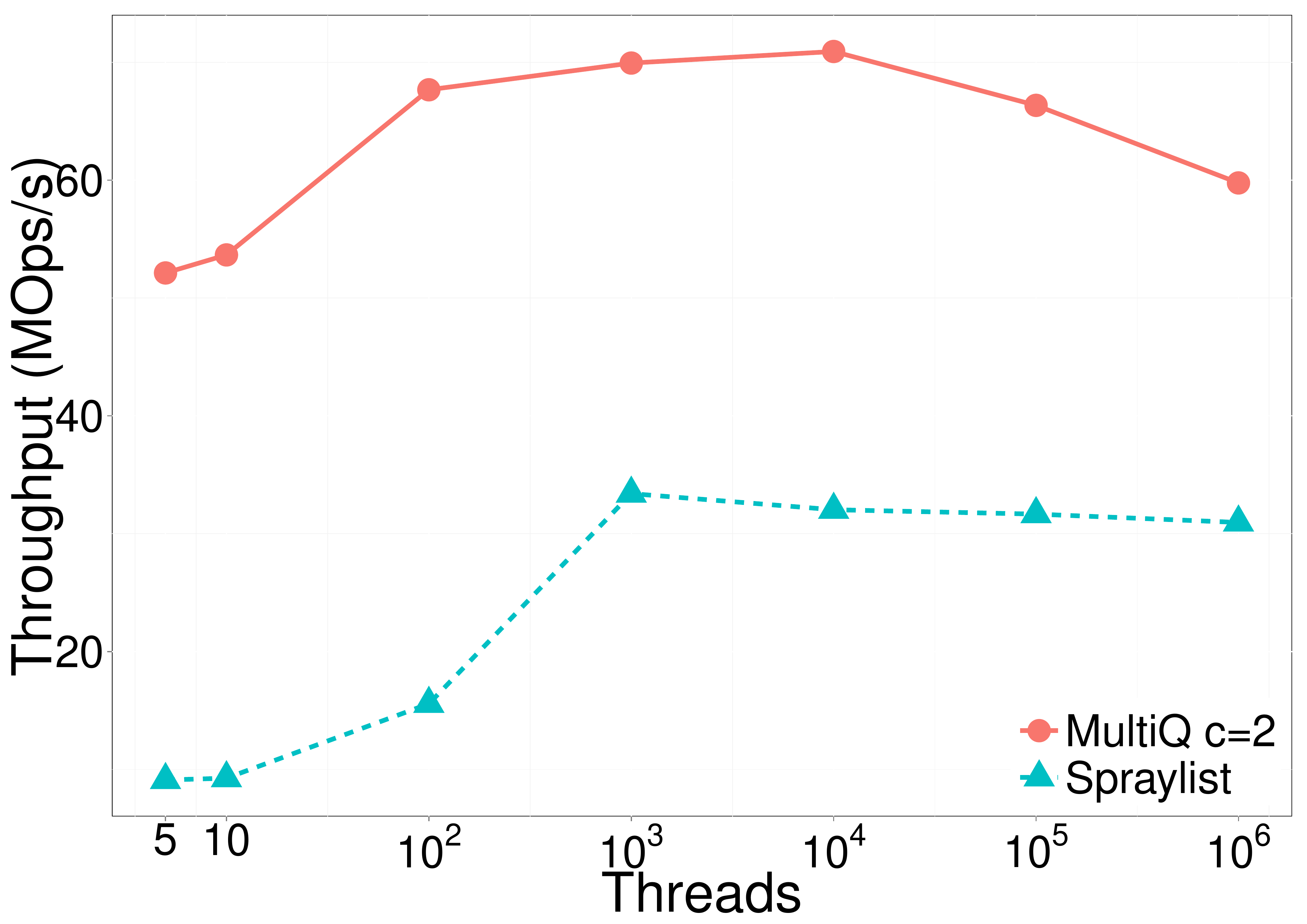}
\caption{Throughput with varying initial size of the priority queue (56 threads)}
\label{fig:size}
\end{figure}

Throughput measurements are run for 1 second and we determine the total number of queue operations completed by all the threads together.
Using this setup, Figure~\ref{fig:throughput} compares
MultiQueues to existing concurrent priority queues which are all based on SkipLists. 
We do not include the queue by Wimmer et al. \cite{WVTCT14} since in Alistarh et al. \cite{AKLS14} its performance is even lower than the SkipList based queues.
Lotan's priority queue \cite{ShaLot00} is a SkipList based priority queue with a lock-free implementation. A \Id{deleteMin} operation will simply
traverse the bottom level and remove the first non-deleted element by setting a marker flag. Physical deletion is done
with a garbage collector. Linden's priority queue \cite{LinJon13} adds some optimizations to minimize CAS operations and
batch the physical deletion of marked nodes. 
These implementations and the implementation of SprayLists are available on the Internet\footnote{\url{github.com/jkopinsky/SprayList} version from March 17, 2014.}. It should be noted however, that the SprayList implementation is not complete since it never cleans up deleted elements. A scalable cleanup component for SprayLists seems to be an open problem.

Among the SkipList based implementations, the SprayList is by far the best except for small number of cores where the Linden queue works better. 

Henceforth, we therefore focus on a comparison with SprayLists.
Multiqueues scale linearly on a single socket ($\leq 14$ threads) being up to 1.7 times faster (for $c=2$) and twice faster (for $c=4$) than SprayList when all cores of the first socket are used. Then there is a dip in performance when going to the second socket, in particular for the case $c=2$ when there is a significant number of failed locking attempts. Using hyperthreading on all cores of a single socket, the overall performance at $p=28$ reaches 1.7 times the value for a single socket (curve label ``HT'').
For $p=56$, MultiQueues are about 1.9 times faster than SprayLists for $c=2$ and about 2.5 times faster for $c=4$. Also recall, that the improvement over a complete implementation of SprayLists is likely to be considerably larger.

\frage{todo: put tic marks on all axes of figures.}
In Figure~\ref{fig:mono} we repeat the measurements for a monotonic distribution of input keys:
After removing an element with key $x$, a thread will insert an element with key $x+y$ where $y$ is a value choosen uniformly at random from $\set{1,\ldots,100}$. Here, SprayLists do not scale at all. Qualitatively, 
it is not surprising that contention increases when inserted elements are close to each other and close to the currently deleted elements. However, it is surprising that performance is collapsing compared to uniformly distributed keys. Note that monotonic behavior of priority queues is typical for many applications, e.g., for discrete event simulation and for finding shortest paths using Dijkstra's algorithm.

In Figure~\ref{fig:size} we show how throughput depends on initial queue size $n_0$. 
For very small sizes, MultiQueues have a sweet spot when average queue sizes are around one, since then the sequential queue access is very fast. For even smaller queues, throughput drops since it takes time to find a nonempty queue. 
SprayLists begin to suffer earlier from small queues  because spraying no longer works -- there are insufficiently many list items available. 

In the Appendix, we give additional performance figures for different hardware.
\subsection{Quality}
\begin{figure}[h]
\centering\includegraphics[width=0.8\linewidth]{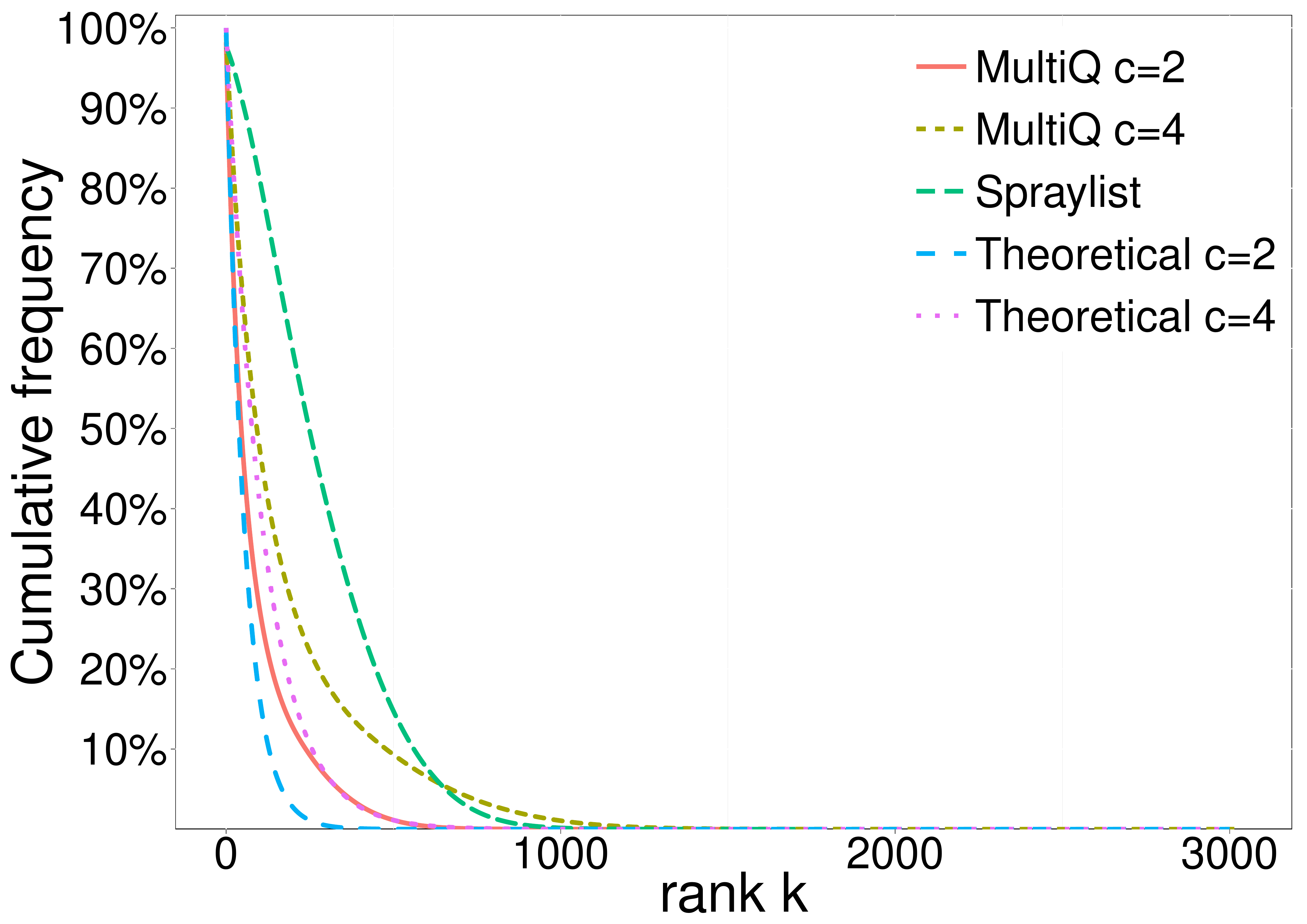}
\caption{Tail distribution of rank error $\prob{\mathrm{rank}>k}$. $10^7$ operations and 56 threads}
\label{fig:quality}
\end{figure}
\begin{figure}[h]
\centering\includegraphics[width=0.8\linewidth]{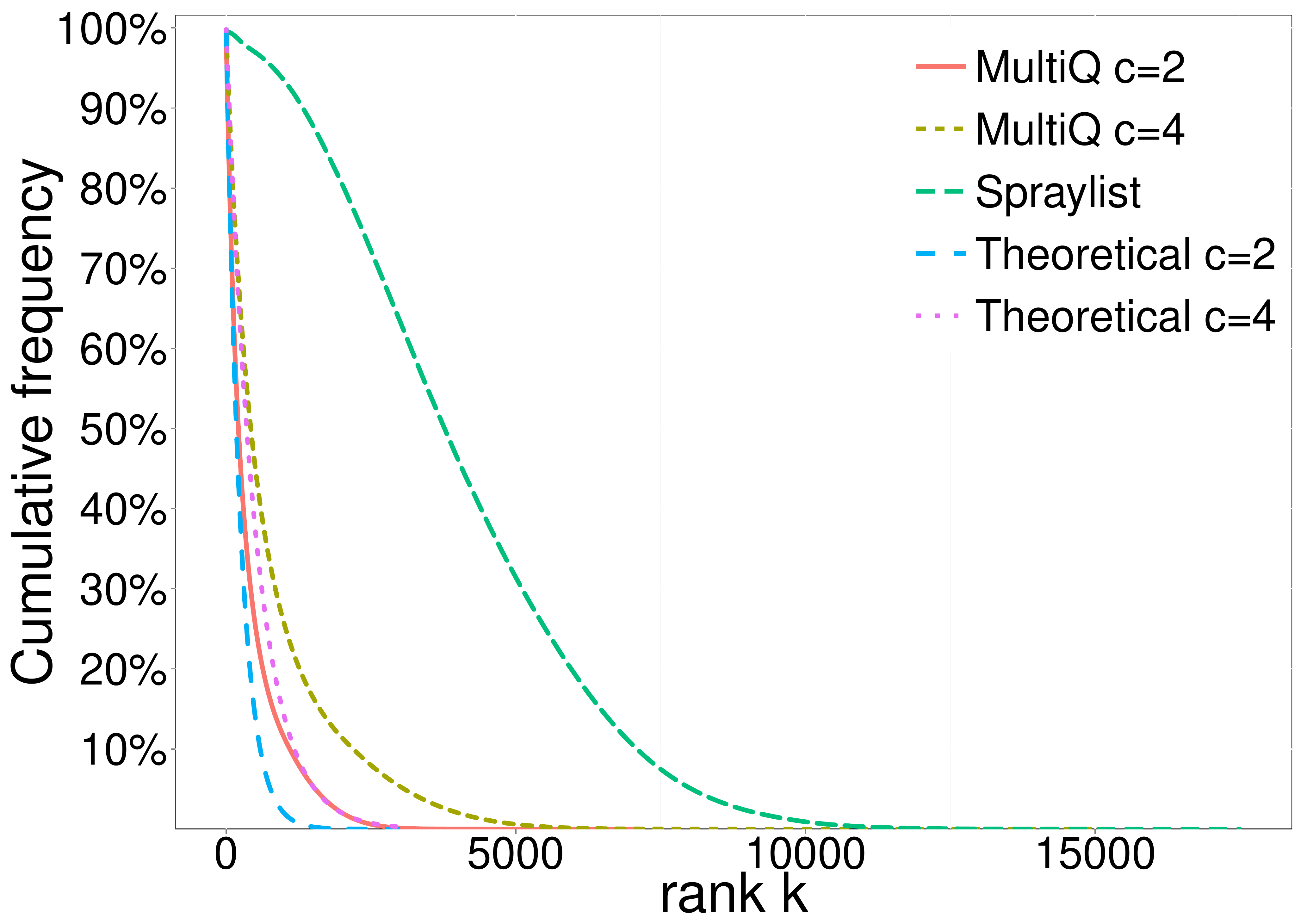}
\caption{Tail distribution of rank error $\prob{\mathrm{rank}>k}$. $10^7$ operations and 256 threads}
\label{fig:quality256}
\end{figure}

\frage{todo: half the x-range of Figure~\ref{fig:quality}. optional: average rank over time. doubleoptional: vary number of sampled queues for \Id{deleteMin}, values 1,2,3,4}

We now measure the distribution of rank errors of MultiQueues and SprayLists. Unfortunately, it is difficult to directly measure this in a running parallel program. We therefore adapt the approach from \cite{AKLS14} to our situation. We run a sequential thread performing the operations which keeps a sorted sequence $S$ containing the current queue elements on the side. After each \Id{deleteMin} returning $x$, the position of $x$ in $S$ gives the rank error.
Figure~\ref{fig:quality} shows the result for $p=56$ and $10^7$ operations (alternating between insertions and deletions). For the ``typically seen'' rank, MultiQueues with $2p=112$ queues ($c=2$) are much better than SprayLists and actually, the measured distribution closely follows the simplified theoretical analysis from Section~\ref{ss:analysis}.
The tails of the distribution are longer however, but still better than for SprayLists.
The largest observed ranks were 1\,522 for MultiQueues and 1\,740 for SprayLists. For $c=4$ the situation is similar, with the exception that the tails of the distribution are worse than for SprayLists.
Figure~\ref{fig:quality256} shows the same experiment for $p=256$.
Again, the distribution follows the theoritical analysis for small ranks.
The advantage over SprayLists grows. Now even MultiQueues with $c=4$ have a distribution with shorter tails.

In the appendix, we give a table with concrete numbers for quartiles of the rank error distribution.


We have also developed a more realistic approach to measuring quality. On a single socket, it is possible to read out a high resolution synchronous clock using the RDTSCP instruction on our test system.
We are recording all operations together with a time stamp. This log can then be used to feed the sequential simulator above -- processing the operations in time-stamp order. We omit the results here since this approach does not scale
to the parameters used in the above measurements. However, the results were generally similar with MultiQueues winning over SprayLists.

\section{Conclusions and Future Work}
\label{s:conclusions}

MultiQueues are a simple way to reduce contention in relaxed priority queues leading to high throughput and good quality. However, there is a need for further research. On the one hand, we would like to have higher quality in particular for the tails of the rank distribution together with provable performance guarantees. We believe that this is an interesting open problem for theoretical research: Devise a protocol for accessing multiple priority queues 
that remains simple and fast but comes with provable quality guarantees.
Parallel priority queues based on multiple queues suggest that this should be possible \cite{San98a} but we want to avoid the frequent global synchronization and rather large implicit constant factors in that approach.

From a practical point of view, we would like to make MultiQueues more scalable over multiple sockets. Possible issues are expensive memory accesses on remote NUMA-nodes and invalidation overheads for cache lines accessed during local queue operations. With respect to this, heaps are not ideal, since a \Id{deleteMin} modifies a logarithmic number of memory locations. Hence, we might try data structure with constant average or amortized update cost. SkipLists might be one option but from our experience \cite{DKMS04} BTrees might be better if we can tailor them for the particular access pattern found in priority queues.


MultiQueues have tunable quality via the number $cp$ of queues. 
It is interesting to consider what happens when we choose $c<1$.
In applications where queue accesses constitute less than a fraction $c$ of the total work, this is likely to work well giving us high quality almost for free. However, if a high operation rate is attempted, $c<1$ will lead to high overhead due to failed locking attempts.  In that case, it would help to introduce a backoff period after failed locking attempts.

Actually, similar problems seem to apply to other relaxed priority queues as well (or even other relaxed data structures like relaxed FIFOs): The application has to specify the degree of relaxation up-front. This has to be done in a rather conservative way since allowing too little relaxation will degrade performance. Another interesting challenge is to design more adaptive data structures that automatically 
and dynamically find the smallest amount of relaxation needed to keep contention low. For example, one could make MultiQueues adaptive by monitoring the operation rate and periodically resize the queue array $Q$ accordingly. 
However, implemented naively, this may be too slow to adapt -- in particular if the queues are large. 

Our technique for wait-free locking may also be applicable in other situations, e.g. for other relaxed data structures. For example, when we apply our technique to relaxed FIFOs, we get a data structure similar to the DQs by Haas et al. \cite{HLHPSKS13} but we can use any sequential FIFO-queue for the individual queues.

\section*{Acknowledgements}
We would like to thank the SAP HANA group for bringing the three authors together,
and Jean-Fran\c{c}ois M\'ehaut for providing a testing platform.

\bibliographystyle{abbrv}
\bibliography{diss,references}
\clearpage
\appendix
\section{Throughput on Other Platforms}
\begin{figure}[h]
\centering\includegraphics[width=0.8\linewidth]{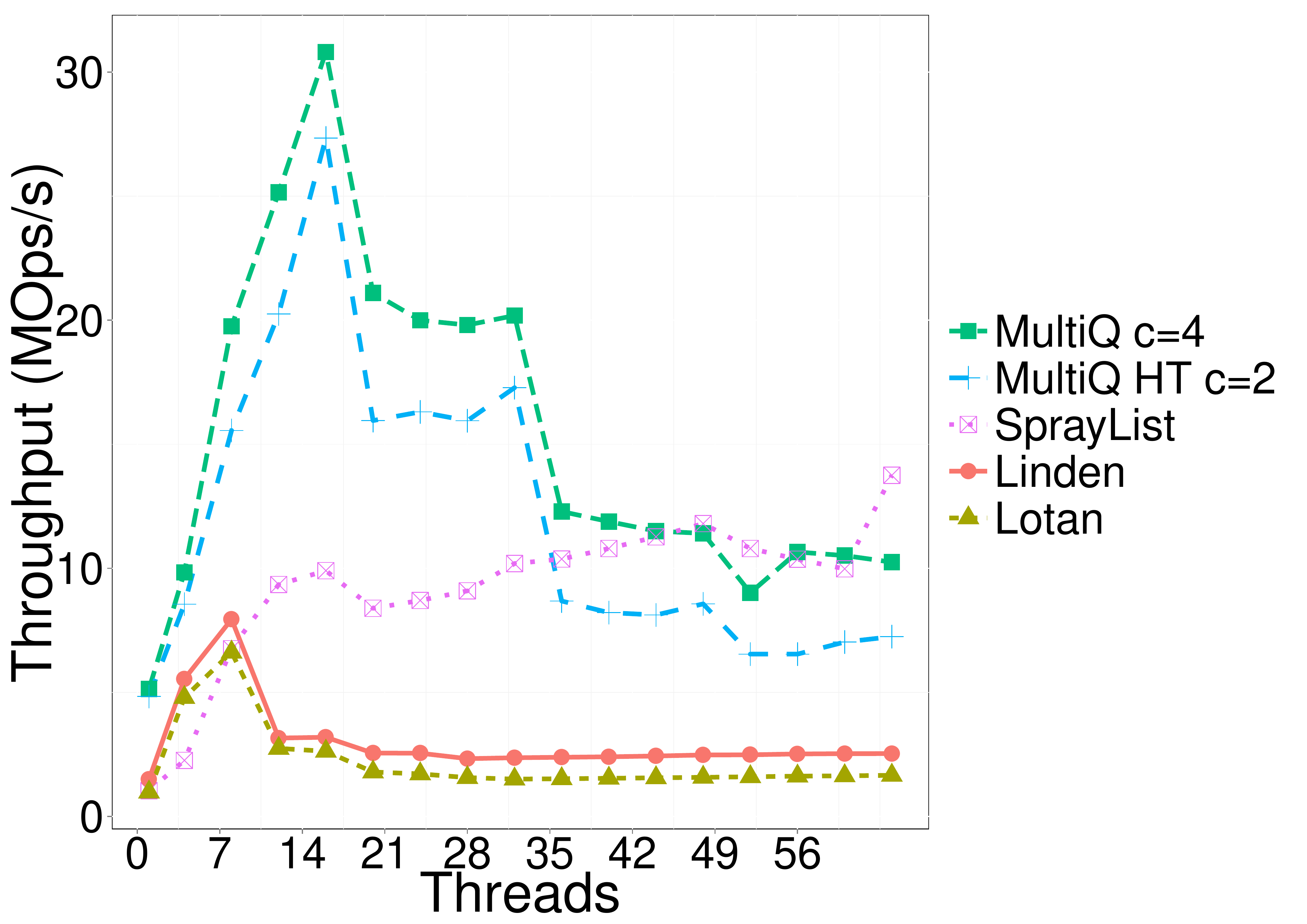}
\caption{Throughput of 50\% \Id{insert} 50\% \Id{deleteMin} operations on a 4 sockets Intel machine.} 
\label{fig:socket4ht}
\end{figure}
\begin{figure}[h]
\centering\includegraphics[width=0.8\linewidth]{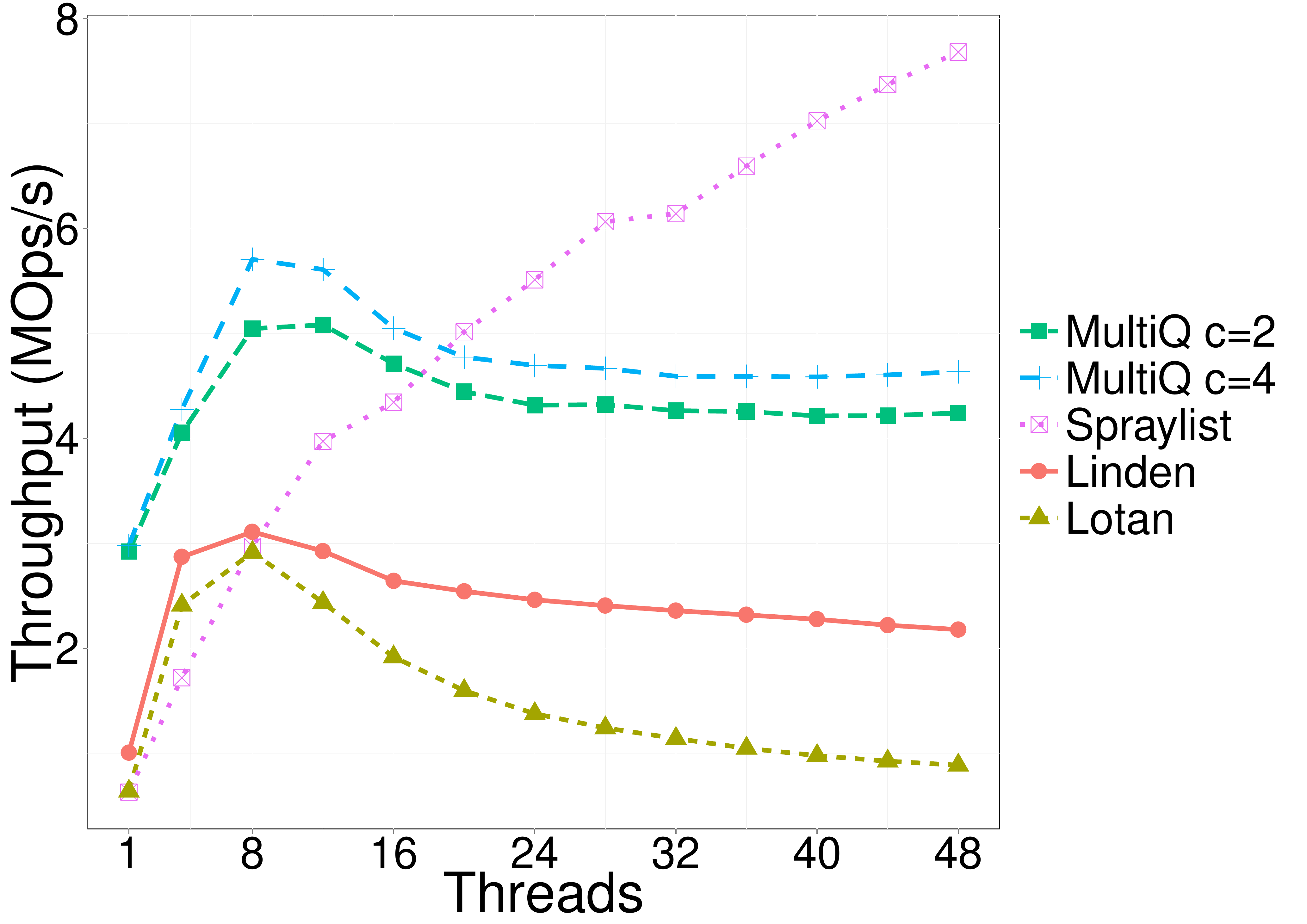}
\caption{Throughput of 50\% \Id{insert} 50\% \Id{deleteMin} operations on a 4 sockets AMD machine} 
\label{fig:4socketAMD}
\end{figure}

Figure~\ref{fig:socket4ht} shows performance for a machine with 
four Intel Xeon E5-4640 processors with 2.4 GHz. Each processor has eight cores with two hardware threads each resulting in 64 hardware threads overall.
The four processors of this machine are connected by single QPI links into a ring network.
For MultiQueues, we first use all hardware threads in a socket before going to the next one. With this configuration, scalability is very good for a single socket. The performance drop for the second socket is considerable. Note that this drop is actually much bigger than for a two-socket machine since only a single QPI link connects two neighboring sockets of our machine whereas a two-socket machine has two parallel QPI links. Once we go to three or four sockets, performance degrades even further. In this configuration the (only partially implemented) SprayList is actually faster.
Figure~\ref{fig:4socketAMD} shows similar on a machine with four AMD Opteron 6168 sockets with 1.9 GHz and 12 cores per socket which confirms the previous results. The performance drops once we go off the first socket.
It is also worth comparing the absolute performance of this machine from 
2010 with the 2014 machine from Figure~\ref{fig:throughput}. The latter has only
6 \% more transistors ($2\times 3.84\cdot 10^9$ versus $4\times 1.8\cdot 10^9$) 16 \% more cores and 37 \% faster clock but, overall, performance is an order of magnitude better. 
\section{Quality}
Table~\ref{tab:quality} gives more details about the rank error distribution. For 10 million \Id{deleteMin} operations and with 56 threads, 75\% of deleted elements have ranks $<= 111$ for $c=2$ and $<=223$ for $c=4$ which is almost 4x and 2x respectively smaller than Spraylist. The highest rank obtained is smaller in the case of MultiQueues with $c=2$ compared with Spraylist but it is greater in the case of $c=4$. However, for 256 threads, the rank error distribution for Spraylist increases faster than MultiQueues.
\begin{table}[h]
\centering
  \begin{tabular}{ l | r   r   r   r  r| r}
    \hline
 \ & 0\% & 25\% & 50\% & 75\% & 100\% & $p$\\ 
    \hline
MultiQueues( c= 2) & 0 & 17  & 46 & 111 & 1\,522 & 56 \\
MultiQueues (c=4) & 0 & 34 & 92 & 223 & 3\,035 & 56\\
Spraylist & 0 & 129 & 249 & 406 & 1\,740 & 56\\\hline
MultiQueues( c= 2) & 0 & 79 & 213 & 513 & 7\,043 & 256\\

MultiQueues (c=4) & 0 & 158 & 427 & 1028 & 15\,224 & 256\\

Spraylist & 0 & 2338 & 3771 & 5512 & 17\,524 & 256\\
    \hline
  \end{tabular}
\caption{Quartiles for the distribution function of rank errors. 10 million operations} \label{tab:quality56} 
\label{tab:quality}
\end{table}
\end{document}

%% file: makros2.tex
\newcommand{\Id}[1]{\ensuremath{\text{{\sf #1}}}}

\newcommand{\set}[1]{\left\{ #1\right\}}

\newcommand{\realrange}[2]{\left[#1, #2\right]}

\newcommand{\unitrange}[2]{\realrange{0}{1}}

\newcommand{\prob}[1]{{\mathbf{P}}\left[#1\right]}

\newcommand{\Oh}[1]{\mathcal{O}\!\left( #1\right)}

\newcommand{\llabel}[1]{\label{\labelprefix:#1}}
\newcommand{\labelprefix}{} 

\newcommand{\discussionsize}{\small}

\newcommand{\frage}[1]{{[\bf #1]}\marginpar{$\bigotimes$}}

\newcommand{\punkt}{\enspace .}

\newenvironment{code}{\noindent\normalsize
\begin{tabbing}%
\hspace{2em}\=\hspace{2em}\=\hspace{2em}\=\hspace{2em}\=\hspace{2em}\=%
\hspace{2em}\=\hspace{2em}\=\hspace{2em}\=\hspace{2em}\=\hspace{2em}\=%
\kill}{\end{tabbing}}

\newcommand{\labelcommand}{}
\newcommand{\captiontext}{}
\newsavebox{\codeparam}
\newcounter{lineNumber}
\newenvironment{disscodepos}[3]{%
\renewcommand{\labelcommand}{#2}%
\renewcommand{\captiontext}{#3}%
\sbox{\codeparam}{\parbox{\textwidth}{#3}}%
\begin{figure}[#1]\begin{center}\begin{code}\setcounter{lineNumber}{1}}{%
\end{code}\end{center}\caption{\llabel{\labelcommand}\captiontext}\end{figure}}

{\end{disscodepos}}

\newcommand{\Procedure}{{\bf Procedure\ }}

\newcommand{\Repeat}   {{\bf repeat\ }}
\newcommand{\Until}    {{\bf until\ }}

\newcommand{\Is}{\mbox{\rm := }}

\newcommand{\If}       {{\bf if\ }}

\newcommand{\Then}     {{\bf then\ }}

\newcommand{\Return}   {{\bf return\ }}

\newcommand{\RRem}[1]   {\`{\bf --\hspace{0.5mm}--~}{\rm#1}}

\newdimen\endofsize\endofsize=0.5em

%% file: tr.bbl
\begin{thebibliography}{10}

\bibitem{AKLS14}
D.~Alistarh, J.~Kopinsky, J.~Li, and N.~Shavit.
\newblock The {SprayList}: A scalable relaxed priority queue.
\newblock Technical Report MSR-TR-2014-16, Microsoft Research, September 2014.

\bibitem{BCSV00}
P.~Berenbrink, A.~Czumaj, A.~Steger, and B.~V{\"o}cking.
\newblock Balanced allocations: The heavily loaded case.
\newblock In {\em 32th Annual {ACM} Symposium on Theory of Computing}, pages
  745--754, 2000.

\bibitem{CMH14}
I.~Calciu, H.~Mendes, and M.~Herlihy.
\newblock The adaptive priority queue with elimination and combining.
\newblock {\em CoRR}, abs/1408.1021, 2014.

\bibitem{DKMS04}
R.~Dementiev, L.~Kettner, J.~Mehnert, and P.~Sanders.
\newblock Engineering a sorted list data structure for 32 bit keys.
\newblock In {\em 6th Workshop on Algorithm Engineering {\&} Experiments}, New
  Orleans, 2004.

\bibitem{DeoPra92}
N.~Deo and S.~Prasad.
\newblock Parallel heap: An optimal parallel priority queue.
\newblock {\em The Journal of Supercomputing}, 6(1):87--98, Mar. 1992.

\bibitem{HLHPSKS13}
A.~Haas, M.~Lippautz, T.~A. Henzinger, H.~Payer, A.~Sokolova, C.~M. Kirsch, and
  A.~Sezgin.
\newblock Distributed queues in shared memory: multicore performance and
  scalability through quantitative relaxation.
\newblock In {\em Proceedings of the ACM International Conference on Computing
  Frontiers}, page~17. ACM, 2013.

\bibitem{HKPSS13}
T.~A. Henzinger, C.~M. Kirsch, H.~Payer, A.~Sezgin, and A.~Sokolova.
\newblock Quantitative relaxation of concurrent data structures.
\newblock In {\em ACM SIGPLAN Notices}, volume~48, pages 317--328. ACM, 2013.

\bibitem{KarZha93}
R.~M. Karp and Y.~Zhang.
\newblock Parallel algorithms for backtrack search and branch-and-bound.
\newblock {\em Journal of the ACM}, 40(3):765--789, 1993.

\bibitem{LaMLad97a}
A.~LaMarca and R.~E. Ladner.
\newblock The influence of caches on the performance of heaps.
\newblock {\em ACM Journal of Experimental Algorithmics}, 1(4), 1996.

\bibitem{LinJon13}
J.~Lind{\'e}n and B.~Jonsson.
\newblock A skiplist-based concurrent priority queue with minimal memory
  contention.
\newblock In {\em Principles of Distributed Systems}, pages 206--220. Springer,
  2013.

\bibitem{Pug90}
W.~Pugh.
\newblock Skip lists: {A} probabilistic alternative to balanced trees.
\newblock {\em Communications of the ACM}, 33(6):668--676, 1990.

\bibitem{RanEtAl94}
A.~Ranade, S.~Cheng, E.~Deprit, J.~Jones, and S.~Shih.
\newblock Parallelism and locality in priority queues.
\newblock In {\em Sixth IEEE Symposium on Parallel and Distributed Processing},
  pages 97--103, October 1994.

\bibitem{San98a}
P.~Sanders.
\newblock Randomized priority queues for fast parallel access.
\newblock {\em Journal Parallel and Distributed Computing, Special Issue on
  Parallel and Distributed Data Structures}, 49:86--97, 1998.

\bibitem{San00b}
P.~Sanders.
\newblock Fast priority queues for cached memory.
\newblock {\em ACM Journal of Experimental Algorithmics}, 5, 2000.

\bibitem{SanWas11}
P.~Sanders and J.~Wassenberg.
\newblock Engineering a multi-core radix sort.
\newblock In {\em Euro-Par}, LNCS. Springer, 2011.
\newblock to appear.

\bibitem{schling2011boost}
B.~Schling.
\newblock {\em The boost C++ libraries}.
\newblock Xml Press, 2011.

\bibitem{ShaLot00}
N.~Shavit and I.~Lotan.
\newblock Skiplist-based concurrent priority queues.
\newblock In {\em 14th Int. Parallel and Distributed Processing Symposium
  (IPDPS)}, pages 263--268. IEEE, 2000.

\bibitem{SunTsi03}
H.~Sundell and P.~Tsigas.
\newblock Fast and lock-free concurrent priority queues for multi-thread
  systems.
\newblock In {\em Int. Parallel and Distributed Processing Symposium (IPDPS)},
  pages 11--20. IEEE, 2003.

\bibitem{Emde77}
P.~van Emde~Boas.
\newblock Preserving order in a forest in less than logarithmic time.
\newblock {\em Information Processing Letters}, 6(3):80--82, 1977.

\bibitem{Wil64}
J.~W.~J. Williams.
\newblock Algorithm 232 ({HEAPSORT}).
\newblock {\em Communications of the ACM}, 7:347--348, 1964.

\bibitem{WVTCT14}
M.~Wimmer, F.~Versaci, J.~L. Tr\"{a}ff, D.~Cederman, and P.~Tsigas.
\newblock Data structures for task-based priority scheduling.
\newblock In {\em 19th ACM SIGPLAN Symposium on Principles and Practice of
  Parallel Programming (PPoPP)}, pages 379--380, New York, NY, USA, 2014. ACM.

\end{thebibliography}
